\documentclass[journal, twoside, draftclsnofoot, 12pt, onecolumn]{IEEEtran}
 \makeatletter

\newcommand{\Rmnum}[1]{\expandafter\@slowromancap\romannumeral #1@}
\makeatother
\usepackage{graphicx,cite,epsfig,amssymb,amsmath,multirow,cases}
\usepackage{amsfonts}
\usepackage{mathrsfs}
\usepackage{algorithm}
\usepackage{algorithmic}
\usepackage{float}
\usepackage{bm}
\usepackage{booktabs}
\usepackage{epstopdf}
\usepackage{ulem}
\usepackage{slashbox}
\usepackage{subfigure}
\usepackage[usenames,dvipsnames]{color}
\usepackage[xcolor]{changebar}
\usepackage{color}

\newtheorem{lemma}{{Lemma}}

\newtheorem{proof}{Proof}

\makeatother
\begin{document}
\bibliographystyle{IEEEtran}
\title{Beam-Squint Mitigating for Reconfigurable Intelligent Surface Aided Wideband MmWave Communications}%


\author{Yun~Chen,~Da~Chen,~and~Tao~Jiang
\thanks{Y.~Chen,~D.~Chen,~and~T.~Jiang are with School of Electronic Information and Communications, Huazhong University of Science and Technology, Wuhan 430074, China (e-mail: chen\_yun@hust.edu.cn; chenda@hust.edu.cn; tao.jiang@ieee.org).}}

\maketitle

\begin{abstract}
In this paper, we focus our attention on the mitigation of beam squint for reconfigurable intelligent surface (RIS) aided wideband millimeter wave (mmWave) communications. Due to the intrinsic passive property, the phase shifts of all elements in RIS should be the same for all frequencies. However, in the wideband scenario, beam squint induced distinct path phases require designing different phase shifts for different frequencies. The above irreconcilable contradiction will dramatically affect the system performance, considering the RIS usually consists of enormous elements and the bandwidth of wideband mmWave communications may be up to several GHz. Therefore, we propose some novel phase shift design schemes for mitigating the effect of beam squint for both line-of-sight (LoS) and non-Los (NLoS) scenarios. Specifically, for the LoS scenario, we firstly derive the optimal phase shift for each frequency and obtain the common phase shift by maximizing the upper bound of achievable rate. Then, for the NLoS scenario, a mean channel covariance matrix (MCCM) based scheme is proposed by fully exploiting the correlations between both the paths and the subcarriers. Our extensive numerical experiments confirm the effectiveness of the proposed phase shift design schemes.
\end{abstract}


\pagestyle{empty}  
\thispagestyle{empty} 
\section{Introduction}
Millimeter wave (mmWave) communication is one of the most favorable techniques for beyond 5G and the future 6G systems, benefiting from its ample unused spectrum resources \cite{Hemadeh2018,YunChenDCN2020,Heath2016}. However, the high path-loss and blockage-susceptible nature of mmWave signal results in a transmission with extremely short range and high line-of-sight (LoS) dependence. To circumvent the inherent disadvantages of mmWave, typically multiple-input multiple output (MIMO) technologies with large-scale antenna arrays are deployed for providing considerable beam gains, but it inevitably gravely increases the consumption of power and cost \cite{YChenIOT,YChenTCOM1}. Moreover, the issue of blockage induced outage has still not been well addressed.

Recently, as a new type of `array', reconfigurable intelligent surface (RIS) has attracted extensive attentions from academia and industry, where each element of RIS is an almost passive device, such as the phase shifter \cite{Liaskos2018}. RIS is capable of adaptively reflecting the incident signals to the desirable directions with much lower power consumption compared to the traditional antenna array having energy-hungry radio frequency chains \cite{Renzo2019}. Moreover, RIS can be flexibly deployed, such as attached to the facades of buildings, so as to create preferable wireless propagation environments and increase the coverage of mmWave communications.

Most of the previous studies on the design of phase shifts in the RIS focus on the sub-6 GHz or narrow band systems \cite{Huang2018,Wu2019,Yu2019,Han2019}. In \cite{Huang2018}, with the aim of maximizing the spectrum and energy efficiencies, the authors proposed low-complexity algorithms to joint power allocation and phase shift design. In \cite{Wu2019}, a semidefinite relaxation (SDR) based scheme was proposed to solve the phase optimization problem with unit modulus constraints. In \cite{Yu2019}, the fixed point iteration and the manifold optimization based algorithms were proposed to obtain locally optimal solutions. In \cite{Han2019}, the phase shift was designed by phase extraction operation by exploiting only the statistical channel state information (CSI).
Moreover, only a few contributions have considered the wideband scenario. For example, in \cite{Taha2019}, the authors considered a RIS having a few randomly distributed active elements and proposed compressive sensing and deep learning based reflection matrix construction schemes with low training overhead.

However, no work has taken the beam squint into account when designing the phase shifts of RIS. Beam squint represents the spatial direction of a beam changes with the frequency, which results in distinct differences between the path phases on different frequencies. However, the near-passive RIS is applied in the time domain, thus the phase shifts are constrained to be the same for all frequencies, which will impose tremendous performance losses given that the element amount of a RIS is usually enormous and the bandwidths of mmWave communications will be up to several GHz.

Against the above background, we dedicate our efforts to mitigating the effect of beam squint for the RIS aided wideband mmWave communications. Both the LoS scenario and the non-LoS (NLoS) scenario between the RIS and the user are considered. Specifically, for the LoS scenario, the optimal phase shift for each frequency is firstly derived. Then, by maximizing the obtained upper bound of the sum-achievable rate, we propose a near-optimal phase shift design scheme which is only based on the long-term angle information. Furthermore, for the NLoS scenario, by fully exploiting the correlations between both the paths and the subcarriers, we propose a mean channel covariance matrix (MCCM) based scheme for obtaining the common phase shift for all frequencies. By means of extensive numerical experiments, we evaluate the performance of the proposed schemes in terms of various system parameters, including the signal-to-noise ratio (SNR), the bandwidth and the number of reflecting elements. Our simulation results confirm the effectiveness of our proposed phase shift design schemes in mitigating the beam squint.

The rest of the paper is organized as follows. Section II presents our system model and channel model. The problem formulation and the proposed phase shift design schemes for both the LoS and NLoS scenarios are elaborated in Section III. Section IV demonstrates our simulation results. Finally, we conclude this paper in Section V.

\emph{Notation}: $ a$ is a scalar, $\bf{a}$ is a vector, and $\bf{A}$ is a matrix. ${\left\| {\bf{a}} \right\|_1}$ and ${\left\| {\bf{a}} \right\|_2}$ denote the $l_1$ and $l_2$ norm of $\bf{a}$, respectively. ${{\bf{A}}^T},{{\bf{A}}^*}$ and $\left| {\bf{A}} \right|$ denote the transpose, conjugate transpose and determinant of ${\bf{A}}$, respectively. ${{\bf{I}}_N}$ is a $N\!\times\! N$ identity matrix. ${\rm diag}(\bf{a})$ denotes a diagonal matrix that consists of the elements of $\bf{a}$. $\mathbb{C}^{M\times N}$ represents the set of all ${M\times N}$ complex-valued matrices. ${\cal C}{\cal N}({{a}},{{b}})$ is a complex Gaussian variable with mean ${{a}}$ and covariance ${{b}}$. ${[{\bf{A}}]_{i,j}}$ is the $(i,j)^{th}$ element of $\bf{A}$.
\section{System Model and Channel Model}

\begin{figure}[t]
\centering
\includegraphics[scale=0.55]{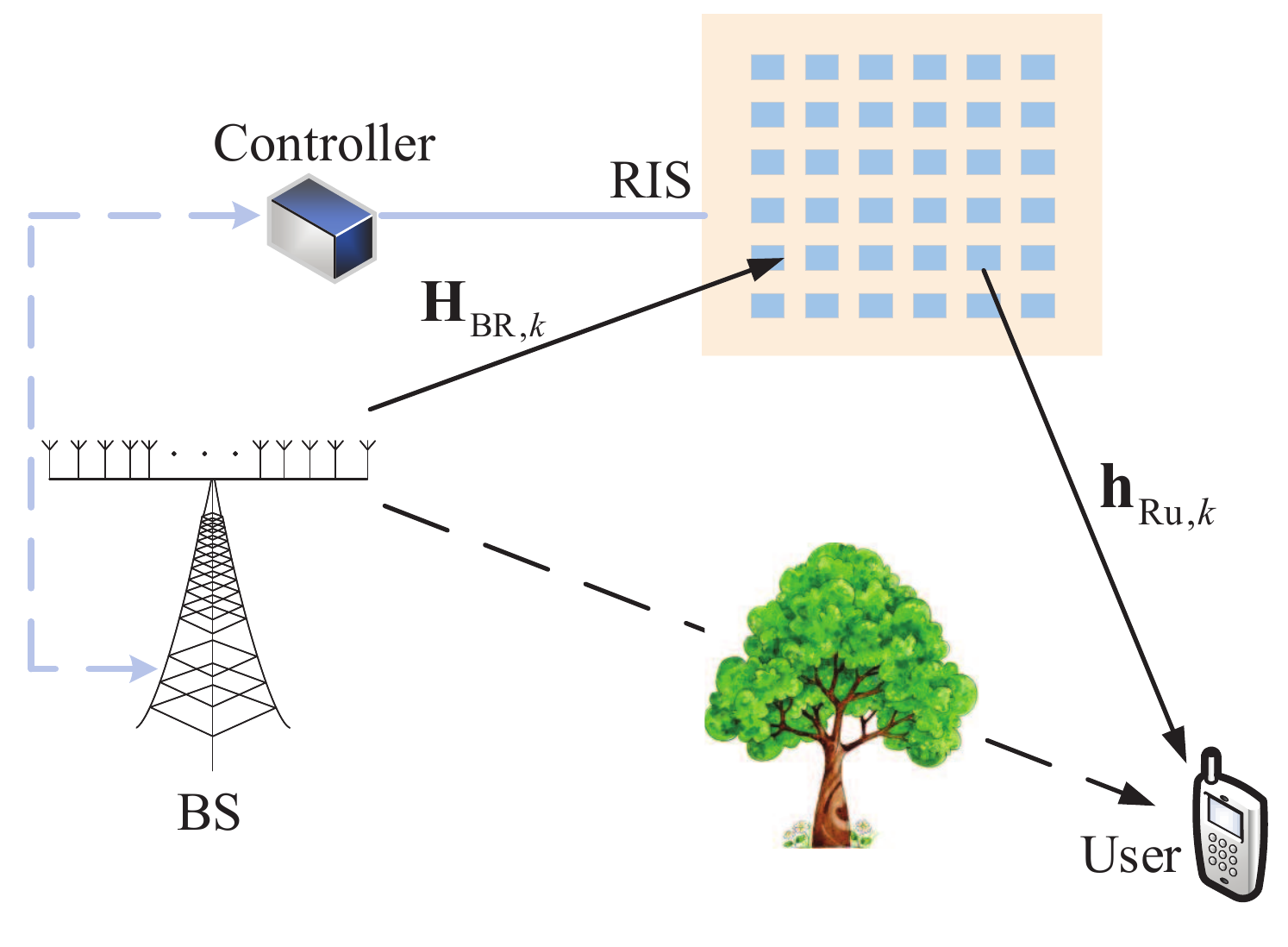}
\caption{A RIS-aided wideband mmWave communication system.}
\label{Fig-systemmodel}
\end{figure}
Consider a RIS-assisted single user wideband mmWave communication system, as shown in Fig. {\ref{Fig-systemmodel}}, where the base station (BS) is equipped with $N$ antennas and transmits signal to a user with the help of $M$-element RIS. To mitigate the dispersion of the wideband mmWave channel, orthogonal frequency-division multiplexing (OFDM) having $K$ subcarriers is utilized to modulate the signal at the BS. The RIS is deployed between the BS and the user, where each element of the RIS is a configurable and programmable phase shifter (PS), whose phase can be dynamically adjusted via a controller between the BS and the RIS. Due to the inherent nature of severe attenuation of mmWave signal, the direct link between the BS and the user is assumed to be negligible and therefore ignored in our system model, similar to \cite{Taha2019}. However, the signal can be reflected by the RIS and reaches the desirable user.  Moreover, the signal reflected by the RIS for more than once is also omitted due to the high path loss.

The received signal of the user at the $k^{th}$ subcarrier can be written as
\begin{equation}\label{2.1}
{y}_{k}=\big({\bf h}_{{\rm Ru},k}{\bf \Phi}{\bf H}_{{\rm BR},k}\big){\bf f}_{k}{s}_{k}+{n}_{k},
\end{equation}
where ${s}_{k}$ denotes the transmitted signal at the $k^{th}$ subcarrier, which satisfies $\mathbb{E}[{s}_k]=0$ and $\mathbb{E}[|{s}_k|^2]=1$,  ${\bf H}_{{\rm BR},k}\in \mathbb{C}^{M\times N}$ and ${\bf h}_{{\rm Ru},k}\in \mathbb{C}^{1\times M}$ are the channels at the $k^{th}$ subcarrier from the BS to the RIS and from the RIS to the user, respectively. Moreover, ${\bf \Phi}={\rm diag}(e^{j\phi_{1}}, e^{j\phi_{2}}, ..., e^{j\phi_{M}})$ denotes the phase shift matrix of the RIS, where $\phi_{m}$ denotes the phase shift of introduced by the $m^{th}$ element. Furthermore, ${\bf f}_{k}\in \mathbb{C}^{M\times 1}$ is the beamforming vector for the $k^{th}$ subcarrier, where we assume perfect channel state information (CSI) is known at the BS and adopt the maximum ratio transmitting (MRT) scheme, similar to \cite{Han2019,Yu2019}, so that we have
\begin{equation}\label{2.2}
{\bf f}_{k}=\sqrt{P}\frac{\big({\bf h}_{{\rm Ru},k}{\bf \Phi}{\bf H}_{{\rm BR},k}\big)^*}{\big\|{\bf h}_{{\rm Ru},k}{\bf \Phi}{\bf H}_{{\rm BR},k}\big\|},
\end{equation}
where $P$ is the total transmit power. Finally, ${n}_{k}\sim {\cal {CN}}(0,\sigma _n^{2})$ denotes the additive white Gaussian noise at the $k^{th}$ subcarrier.

Typically, the mmWave channel possesses quite limited scattering paths due to the severe path loss, the Saleh-Valenzuela model is adopted for modelling the wideband mmWave channel \cite{YChenTVT,YChenTCOM2,Ayach2013}. For the channel matrix between the BS and the RIS, we assume that there is only a LoS path, since the locations of the BS and the RIS are both fixed and it is reasonable to place the RIS at a position where the LoS transmission is feasible. Moreover, for mmWave channel, when the LoS path exists, the NLoS paths is negligible. Therefore, ${\bf H}_{{\rm BR},k}$ can be formulated as
\begin{equation}\label{2.}
{\bf H}_{{\rm BR},k}=\gamma_{\rm BR}{{\alpha _{{\rm BR}}}}{\beta_{{\rm BR},k}}{{\bf{a}}_{\rm{M}}}(\phi _{{\rm BR},k}){{\bf{a}}_{\rm{N}}^*}(\varphi_{{\rm BR},k}),
\end{equation}
where $\gamma_{\rm BR}=\sqrt{MN}$ is the normalization factor, ${{\alpha _{{\rm BR}}}}$ is the path gain, and ${\beta_{{\rm BR},k}}= e^{-j2\pi\tau_{}f_{k}}$ denotes the delay component, in which $\tau_{} \sim {\cal {U}}(0,20{\rm ns})$ is the delay of the path. Moreover, ${{\bf{a}}_{\rm{M}}}(\phi _{{\rm BR},k})$ and ${{\bf{a}}_{\rm{N}}}(\varphi_{{\rm BR},k})$ are the array response vectors for the uniform linear array (ULA). Please note that the proposed schemes are applicable to arbitrary antenna arrays, such as uniform planar arrays (UPA). Therefore, the $N$-element array response vector can be written as
\begin{equation}\label{2.3}
{{\bf{a}}_{\rm{N}}}(\phi_k)=\dfrac{1}{{\sqrt {{N}} }}\Big[1, e^{j2\pi\phi_k},..., e^{j({N-1})2\pi\phi_k}\Big]^T,
\end{equation}
where $\phi _{k}$ is the spatial angle of the path at the $k^{th}$ subcarrier, which can be written as
\begin{equation}\label{2.5}
\phi _{k}=\frac{f_k}{\rm c}d{\rm sin}(\theta),
\end{equation}
where $f_{k}=f_{\rm c}+\frac{f_{\rm s}}{K}(k-1-\frac{K-1}{2})$ is the frequency at the $k^{th}$ subcarrier, $f_{\rm c}$ and $f_{\rm s}$ are the carrier frequency and bandwidth, respectively, while $\rm c$ represents the speed of light and $d=\frac{\rm c}{2f_{\rm c}}=\frac{\lambda}{2}$ is the antenna-element spacing, in which $\lambda$ is the signal wavelength of the central frequency. Moreover, $\theta\in [0,2\pi)$ denotes the corresponding physical angle. In the narrow-band scenario or wideband scenario with relatively low bandwidth, $f_{k}$ is usually assumed to be equal to $f_{\rm c}$, resulting in that the spatial angles are frequency-independent, i,e., $\phi_{k}=\frac{f_{\rm c}}{\rm c}\frac{\rm c}{2f_{\rm c}}{\rm sin}(\theta_{})=\frac{1}{2}{\rm sin}(\theta_{})$. However, for fulfilling the requirement of future extremely high speed transmission, utilizing  substantial bandwidth resources of the mmWave band or of the higher frequency band is essential. Therefore, in this case, the assumption $f_{k}=f_{\rm c}$ no longer holds and the more practical spatial angle $\phi _{k}$ will be frequency-dependent due to the strong beam squint, which will further result in the difference of subspaces of distinct subcarrier channel matrices. However, all the elements of RIS are  near-passive and are applied in the time domain, thus the phase shifts are the same for all subcarriers, which gravely constrains the performance of the RIS aided wideband mmWave communications.

For the channel vector between the RIS and the user, we consider both the LoS and the NLoS scenarios, since the user is usually in mobility, therefore, ${\bf h}_{{\rm Ru},k}$ can be formulated as
\begin{numcases}{\!\!\!\!\!\!\!\! {\bf h}_{{\rm Ru},k} =}
\gamma_{\rm Ru}{\alpha _{{\rm Ru}}}{\beta_{{\rm Ru},k}}{{\bf{a}}_{\rm{M}}^*}(\varphi_{{\rm Ru},k}), &\!\!\!\!\!\!\! {\rm LoS}, \\
\tilde{\gamma}_{\rm Ru}\sum\limits_{l=1}^{L}{{\alpha _{{\rm Ru},l}}}{\beta_{{\rm Ru},l,k}}{{\bf{a}}_{\rm{M}}^*}(\varphi_{{\rm Ru},l,k}), &\!\!\!\!\!\!\! {\rm NLoS}.
\end{numcases}
where $\gamma_{\rm Ru}=\sqrt{M}$ and $\tilde{\gamma}_{\rm Ru}=\sqrt{M/L}$ are the corresponding normalization factors, ${{\alpha _{{\rm Ru}}}}$ and ${{\alpha _{{\rm Ru},l}}}$ denote the path gains, while ${\beta_{{\rm Ru},k}}$ and ${\beta_{{\rm Ru},l,k}}$ denote the delay components on the $k^{th}$ subcarrier for the LoS scenario and the NLoS scenario, respectively. Similarly, the channel vector ${\bf h}_{{\rm Ru},k}$ is also dramatically affected by the beam squint.


\section{Phase Shift Design for Wideband MmWave Communications}

\subsection{Problem Formulation}
The objective of designing the phase shifts in the RIS is to maximize the achievable sum rate $R_{\rm {sum}}$ at the user over the wideband mmWave channel. According to (\ref{2.1}) and (\ref{2.2}), $R_{\rm {sum}}$ is given by
\begin{equation}\label{3.1}
\begin{array}{l}
R_{\rm {sum}}=\dfrac{1}{K}\sum\limits_{k=1}^{K}{\log _2}\Big( 1+\frac{P}{\sigma _n^{2}}\big\|{\bf h}_{{\rm Ru},k}{\bf \Phi}{\bf H}_{{\rm BR},k} \big\|^2 \Big ).
\end{array}
\end{equation}
Therefore, our optimization problem can be formulated as
\begin{equation}\label{3.}
\begin{array}{l}
{\bf \Phi}^{\rm opt} = \mathop {{\rm{arg}}{\kern 1pt} {\kern 1pt} {\rm{max}}} R_{\rm {sum}}, \forall k,\\
 {\kern 35pt} {\rm{s}}{\rm{.t}}{\rm{.}}{\kern 7pt}  \big|[{\bf \Phi}]_{m,m}\big|=1, \forall m.
\end{array}
\end{equation}
Under the constant modulus constraint on the elements of ${\bf \Phi}$, maximizing the above achievable rate is non-convex and thus intractable, which requires highly sophisticated iteration search process, such as the fixed point iteration and the manifold optimization. Moreover, as has been mentioned in Section II, ${\bf \Phi}$ is the same for all subcarriers, which makes maximizing the achievable rate even more challenging.
Instead, in this paper, we mainly focus on low-complexity but effective solutions, which are only based on the long-term CSI, and try to provide some designing insights for RIS based wideband mmWave communications in the face of beam squint.
Specifically, in the following two subsections, for both the LoS scenario and the NLoS scenario between the RIS and the user, the near-optimal RIS phase shift design schemes are proposed based on the long-term angle information and the MCCM, respectively.

\subsection{LoS Scenario}
In this section, we assume that there is only one LoS path between the RIS and the user, which is actually the most common and practical scenario for RIS assisted mmWave communications. In this case, the channel vector between the RIS and the user is shown as (6). Though the LoS scenarios is simpler than other scenarios having multiple paths, obtaining the optimal phase shift matrix for the wideband mmWave communication systems is still challenging. Let us kick off from the following lemma.

\begin{lemma}
For achievable rate on the $k^{th}$ subcarrier channel $R_k$, the optimal phase shift of the $m^{th}$ element in the RIS can be designed as
\begin{equation}\label{3.}
\phi_{m}^{\rm opt}=2\pi(m-1)(-\phi _{{\rm BR},k}+\varphi_{{\rm Ru},k}),
\end{equation}
and $\{{\rm{max}}{\kern 1pt} {\kern 1pt} R_k\}$ is equivalent to $\{{\rm{max}}{\kern 1pt} {\kern 1pt} \big | z_k \big |^2\}$, and is further equivalent to
\begin{equation}\label{3.10}
\mathop { {\rm{min}}} \big| 2\pi(m-1)(-\phi _{{\rm BR},k}+\varphi_{{\rm Ru},k})-\phi_{m} \big |, \forall m.
\end{equation}
\end{lemma}
\begin{proof}
According to (\ref{3.1}), $R_k$ can be written as
\begin{equation}\label{3.}
\begin{split}
R_k&={\log _2}\Big( 1+\frac{P}{\sigma _n^{2}}\big\|{\bf h}_{{\rm Ru},k}{\bf \Phi}{\bf H}_{{\rm BR},k} \big\|^2 \Big )\\
&={\log _2}\Big( 1+\frac{P}{\sigma _n^{2}}\big | z_k \big |^2 \|{{\bf{a}}_{\rm{N}}^*}(\varphi_{{\rm BR},k})\|^2\Big )\\
&={\log _2}\Big( 1+\frac{PN}{\sigma _n^{2}}\big | z_k \big |^2 \Big ),
\end{split}
\end{equation}
where $z_k=\sum\limits_{m=1}^{M}z_{k,m}=\sum\limits_{m=1}^{M}e^{j2\pi(m-1)(\phi _{{\rm BR},k}-\varphi_{{\rm Ru},k})+j\phi_{m}}$. The optimal phase shift for the $k^{th}$ frequency can be readily obtained by making all elements of $z_{k,m}$ in phase, i.e., $\phi_{m}=2\pi(m-1)(-\phi _{{\rm BR},k}+\varphi_{{\rm Ru},k})+\rm C$, which is the phase difference of the array vectors before and after the RIS reflection. Moreover, without loss of generality, $\rm C$ can be set to 0, which means maximizing $R_k$ is equivalent to making phases of all elements in $z_k$ equal to 0.
\end{proof}

Due to the inherent passive nature of RIS, the phase shift matrix for all subcarriers' channels will be the same. However, the angles of paths are spread from the central frequency when taking the beam squint into account according to (\ref{2.5}). Therefore, it is impossible to guarantee all the elements of  $z_k$ in phase. Nevertheless, a near-optimal solution can be obtained by maximizing the upper bound of the achievable rate $R^{\rm {ub}}$, where $R^{\rm {ub}}$ is given by
\begin{equation}\label{2.}
\begin{split}
R_{\rm {sum}}&=\frac{1}{K}\sum\limits_{k=1}^{K}{\log _2}\Big( 1+\frac{PN}{\sigma _n^{2}}\big | z_k \big |^2 \Big )\\
&\leq{\log _2}\Big( 1+\frac{PN}{\sigma _n^{2}K}\sum\limits_{k=1}^{K}\big | z_k \big |^2 \Big )\\
&\overset{\vartriangle}=R^{\rm {ub}},
\end{split}
\end{equation}
and the problem ${\bf \Phi}^{\rm ub} = \mathop {{\rm{arg}}{\kern 1pt} {\kern 1pt} {\rm{max}}} R^{\rm {ub}}$ is equivalent to
\begin{equation}\label{2.13}
{\bf \Phi}^{\rm ub} = \mathop {{\rm{arg}}{\kern 1pt} {\kern 1pt} {\rm{max}}} \dfrac{1}{K}\sum\limits_{k=1}^{K}\big | z_k \big |^2.
\end{equation}
Direct solution of (\ref{2.13}) is still intractable. However, motivated by (\ref{3.10}), we propose to solve an alternative suboptimal problem, which is formulated as
\begin{equation}\label{2.14}
 {\rm{min}} {\kern 1pt} {\kern 1pt} \dfrac{1}{K}\sum\limits_{k=1}^{K}\big| 2\pi(m-1)(-\phi _{{\rm BR},k}+\varphi_{{\rm Ru},k})-\phi_{m} \big |, \forall m,
\end{equation}
which means the designed phase shift of each element in the RIS is to minimize the sum of phase differences for all subcarriers so as to reflect the incident signal carried by each subcarrier to the desirable direction as close as possible and the solution of (\ref{2.14}) can be readily obtained by
\begin{equation}\label{3.15}
\begin{split}
\phi_{m}^{\rm }&=\dfrac{1}{K}\sum\limits_{k=1}^{K} 2\pi(m-1)(-\phi _{{\rm BR},k}+\varphi_{{\rm Ru},k})\\
&=\dfrac{1}{K}\sum\limits_{k=1}^{K} 2\pi(m-1)\frac{f_k}{\rm c}d({\rm sin}(\vartheta_{\rm Ru})-{\rm sin}(\theta_{\rm BR}))\\
&=2\pi(m-1)\frac{\dfrac{1}{K}\sum\limits_{k=1}^{K}f_k}{\rm c}d({\rm sin}(\vartheta_{\rm Ru})-{\rm sin}(\theta_{\rm BR}))\\
&=2\pi(m-1)\frac{f_{\rm c}}{\rm c}d({\rm sin}(\vartheta_{\rm Ru})-{\rm sin}(\theta_{\rm BR}))\\
&=\pi(m-1)({\rm sin}(\vartheta_{\rm Ru})-{\rm sin}(\theta_{\rm BR})),
\end{split}
\end{equation}
which is actually the phase shift difference on the central frequency. (\ref{3.15}) indicates that even in the face of beam squint, reflecting the incident signal along the path angle  of the central subcarrier is capable of achieving quite good performance, which is also confirmed by the simulation results in Section IV.A. Moreover, constructing the phase shift matrix according to (\ref{3.15}) only requires the long-term angle information ($\vartheta_{\rm Ru}$ and $\theta_{\rm BR}$) and extreme low computational complexity.

\subsection{NLoS Scenario}
Considering the user is usually in mobility, the LoS transmission can not be guaranteed all the time. Therefore, in this subsection, we also consider the more complicated scenario where the LoS path between the RIS and the user is blocked, but the signal can be transmitted through multiple NLoS paths. In this case, the channel vector between the RIS and the user is given by (7).

Recalling the achievable rate from (\ref{3.1}), which obeys the following inequalities
\begin{equation}\label{4.16}
\begin{split}
R_{\rm {sum}}&=\dfrac{1}{K}\sum\limits_{k=1}^{K}{\log _2}\Big( 1+\frac{P}{\sigma _n^{2}}\big({\bf h}_{{\rm Ru},k}{\bf \Phi}{\bf H}_{{\rm BR},k}{\bf H}_{{\rm BR},k}^*{\bf \Phi}^*{\bf h}_{{\rm Ru},k}^* \big) \Big )\\
&\overset{(a)}{\leq}\dfrac{1}{K}\sum\limits_{k=1}^{K}{\log _2}\Big( 1+\frac{PN}{\sigma _n^{2}}\big({\bf h}_{{\rm Ru},k}{\bf \Phi}_{2}{\bf B}_{{\rm M}}{\bf \Phi}_{2}^*{\bf h}_{{\rm Ru},k}^* \big) \Big )\\
&\overset{(b)}=\dfrac{1}{K}\sum\limits_{k=1}^{K}{\log _2}\Big(\Big | {\bf I}_{M}+\frac{PN}{\sigma _n^{2}}\big({\bf B}_{{\rm M}}{\bf \Phi}_{2}^*{\bf h}_{{\rm Ru},k}^*{\bf h}_{{\rm Ru},k}{\bf \Phi}_{2}\big) \Big |\Big)\\
&\overset{(c)}{\leq}{\log _2}\Big(\Big | {\bf I}_{M}+\frac{PN}{\sigma _n^{2}}\big({\bf B}_{{\rm M}}{\bf \Phi}_{2}^*\frac{1}{K}\sum\limits_{k=1}^{K}{\bf h}_{{\rm Ru},k}^*{\bf h}_{{\rm Ru},k}{\bf \Phi}_{2}\big) \Big |\Big),
\end{split}
\end{equation}
where ${\bf \Phi}$ is divided into two phase shift matrices as ${\bf \Phi}={\bf \Phi}_{2}{\bf \Phi}_{1}$, in which ${\bf \Phi}_{1}$ is to receive the incident signal, while ${\bf \Phi}_{2}$ is to forward the signal to the user. For the LoS channel between the BS and the RIS, if ${\bf \Phi}_{1}$ can perfectly receive the incident signal from each subcarrier channel, we have ${\bf \Phi}_{1}{{\bf{a}}_{\rm{M}}}(\phi _{{\rm BR},k})={1}/{\sqrt{M}}[\underbrace{1,1,...,1}_M]^T, k=1,2,...,K$ and thus ${\bf \Phi}_{1}{\bf H}_{{\rm BR},k}{\bf H}_{{\rm BR},k}^*{\bf \Phi}_1^*=N{\bf B}_{{\rm M}}$, where ${\bf B}_{{\rm M}}$ is a ${\rm M\times M}$ matrix with all elements are 1. However, since ${\bf \Phi}_{1}$ is constrained by the passive characteristic of RIS and should be the same for all frequencies, we obtain the inequality (a). (b) is due to the fact that $|\bf I + AB | = |I + BA|$ by defining ${\bf A} = {\bf h}_{{\rm Ru},k}{\bf \Phi}_{2}$ and ${\bf B} = {\bf B}_{{\rm M}}{\bf \Phi}_{2}^*{\bf h}_{{\rm Ru},k}^*$, Finally, (c) is derived by employing Jensen's inequality.

Note that the inner item in the last line of (\ref{4.16}) $\frac{1}{K}\sum\limits_{k=1}^{K}{\bf h}_{{\rm Ru},k}^*{\bf h}_{{\rm Ru},k}$ is actually the MCCM between the RIS and the user. Moreover, according to (\ref{4.16}), the optimal unconstrained ${\bf \Phi}_{2}$ can be constructed by using the first column of the singular matrix of $\frac{1}{K}\sum\limits_{k=1}^{K}{\bf h}_{{\rm Ru},k}^*{\bf h}_{{\rm Ru},k}$. To meet the constant modulus constraint, ${\bf \Phi}_{2}$ is finally given by applying further phase extraction operation. Similar as the method proposed in Section III.B, ${\bf \Phi}_{1}$ is determined by the path phase of the central subcarrier's channel between the BS and the RIS.

The proposed scheme is based on the singular value decomposition (SVD) on the MCCM, which makes full use of the correlations between both the paths and the subcarriers so as to improve the performance. Moreover, the main computational complexity of the proposed scheme also comes from the SVD operation, which is much lower than that of the sophisticated iteration based schemes, such as the manifold optimization based scheme in \cite{Yu2019}.

\section{Simulation Results}
In this section, we provide extensive numerical results for characterizing the performance of the proposed schemes in RIS aided wideband mmWave communications.
Unless otherwise indicated, the simulation parameters are given as follows. The central carrier frequency is $f_{\rm c}=$ 28 GHz, the number of carriers is $K=128$, and the bandwidth considered is $f_{\rm s}=$ 2 GHz. Moreover, the number of antennas at the BS is $N=64$, while the number of elements in the RIS is $M=64$. Finally, the path angles for both LoS and NLoS scenarios are randomly distributed within $(0, 2\pi]$.

\subsection{LoS scenario}
\begin{figure}[t]
\centering
\includegraphics[scale=0.6]{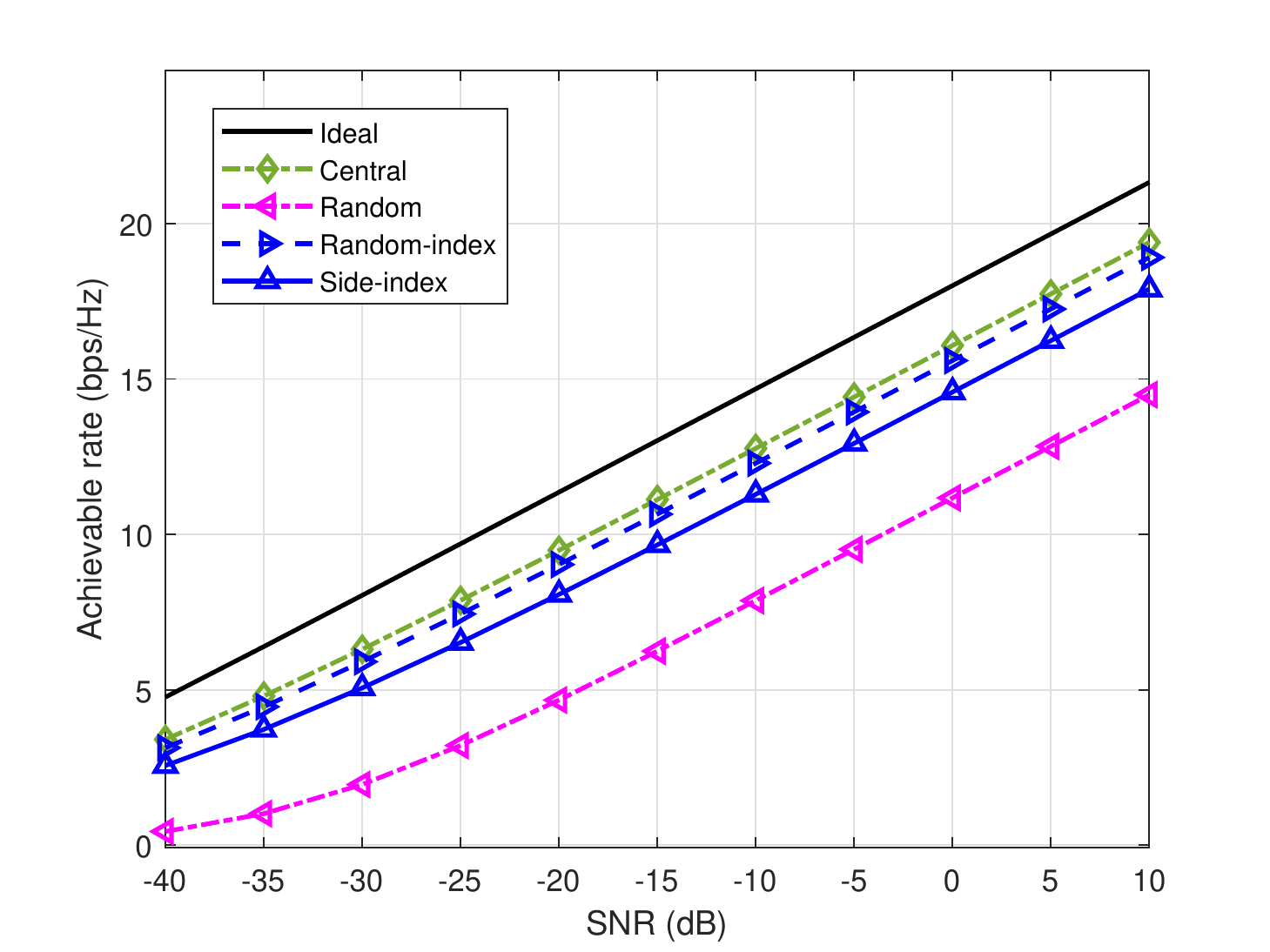}
\caption{Achievable rate vs SNR in the LoS scenario.}
\label{Fig-SNR}
\end{figure}
We firstly show the simulation results for the LoS scenario in this subsection, where the proposed phase shift design scheme shown in (\ref{3.15}) is marked as `Central' in Figs. 2-4. Moreover, we also adopt several benchmarks, in which `Ideal' is the upper bound where the phase shift matrix reflects the signal according to the different phase differences on different subcarriers, `Random' is the scheme in which the phase of each element in the RIS is set randomly, `Random-index' represents the scheme where the phase shifts are determined based on the phase differences of a subcarrier whose index is randomly selected, and finally `Side-index' is the scheme such that the phase shifts are designed according to the phase differences of the subcarrier on either side.

Fig. {\ref{Fig-SNR}} shows the achievable rate for different SNRs. It can be observed that there are non-negligible performance gaps between the `Ideal' scheme and other schemes, which reveals the inherent disadvantage of the passive RIS in the wideband scenario when taking the beam squint into account. However, the `Central' scheme always outperforms the `Random', `Random-index' and `Side-index' schemes, which confirms the effectiveness of our proposed scheme in mitigating the effect of beam squint.

\begin{figure}[t]
\centering
\includegraphics[scale=0.6]{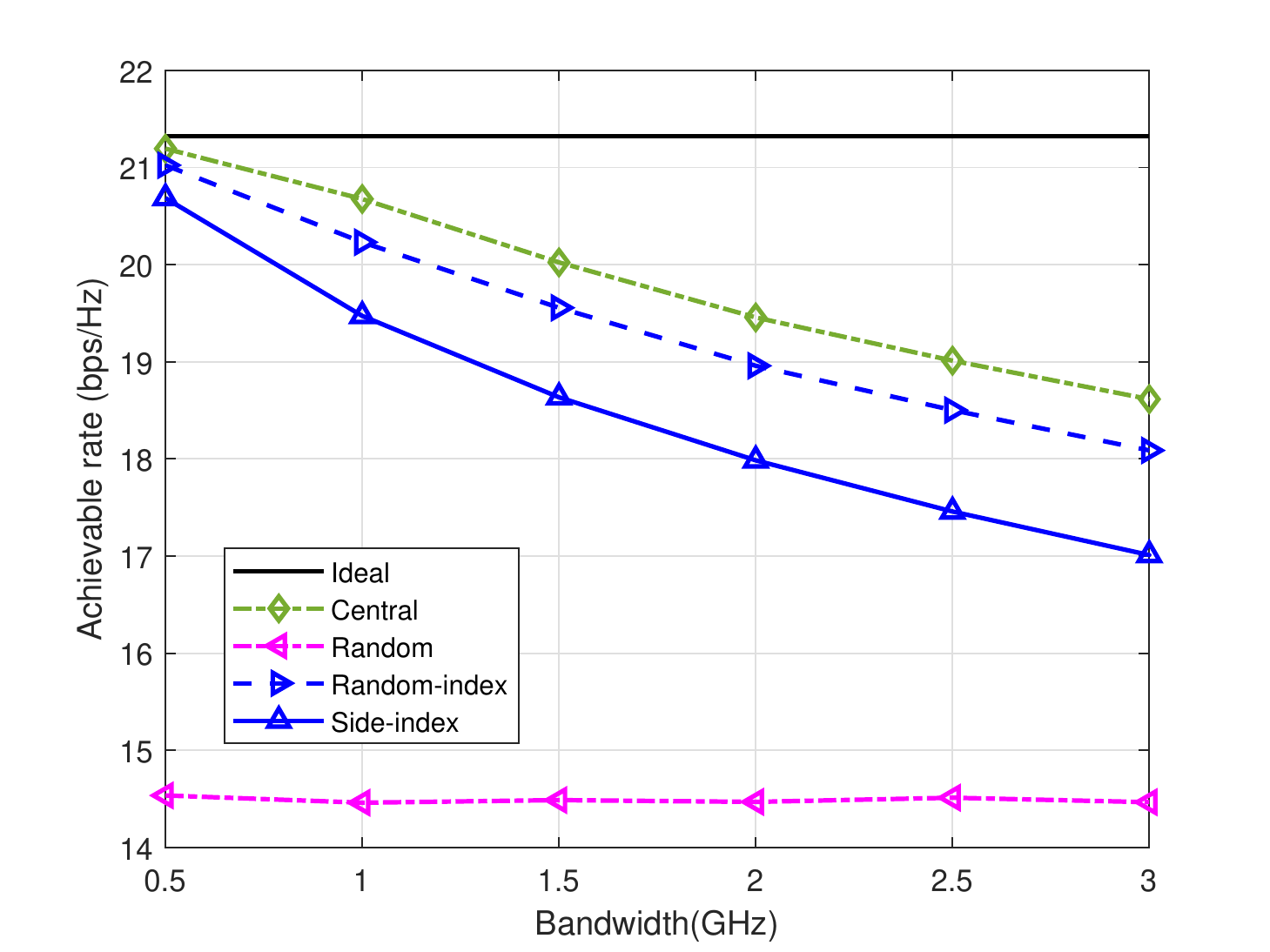}
\caption{Achievable rate vs Bandwidth in the LoS scenario.}
\label{Fig-Bandwidth}
\end{figure}
Fig. {\ref{Fig-Bandwidth}} presents the performance of different schemes for different bandwidths, where the SNR is 10 dB. We observe that the proposed scheme is capable of achieving near-optimal performance when the bandwidth is 500 MHz and attains better performance than `Random-index' and `Side-index' schemes for both low bandwidth and extremely high bandwidths. Moreover, the performance gap between the `Ideal' scheme and the other schemes except the `Random' scheme increases with the bandwidth, which indicates that the beam squint will gravely affect the performance of RIS assisted wideband communications and it is of great significance to mitigate the effect of beam squint.

\begin{figure}[t]
\centering
\includegraphics[scale=0.6]{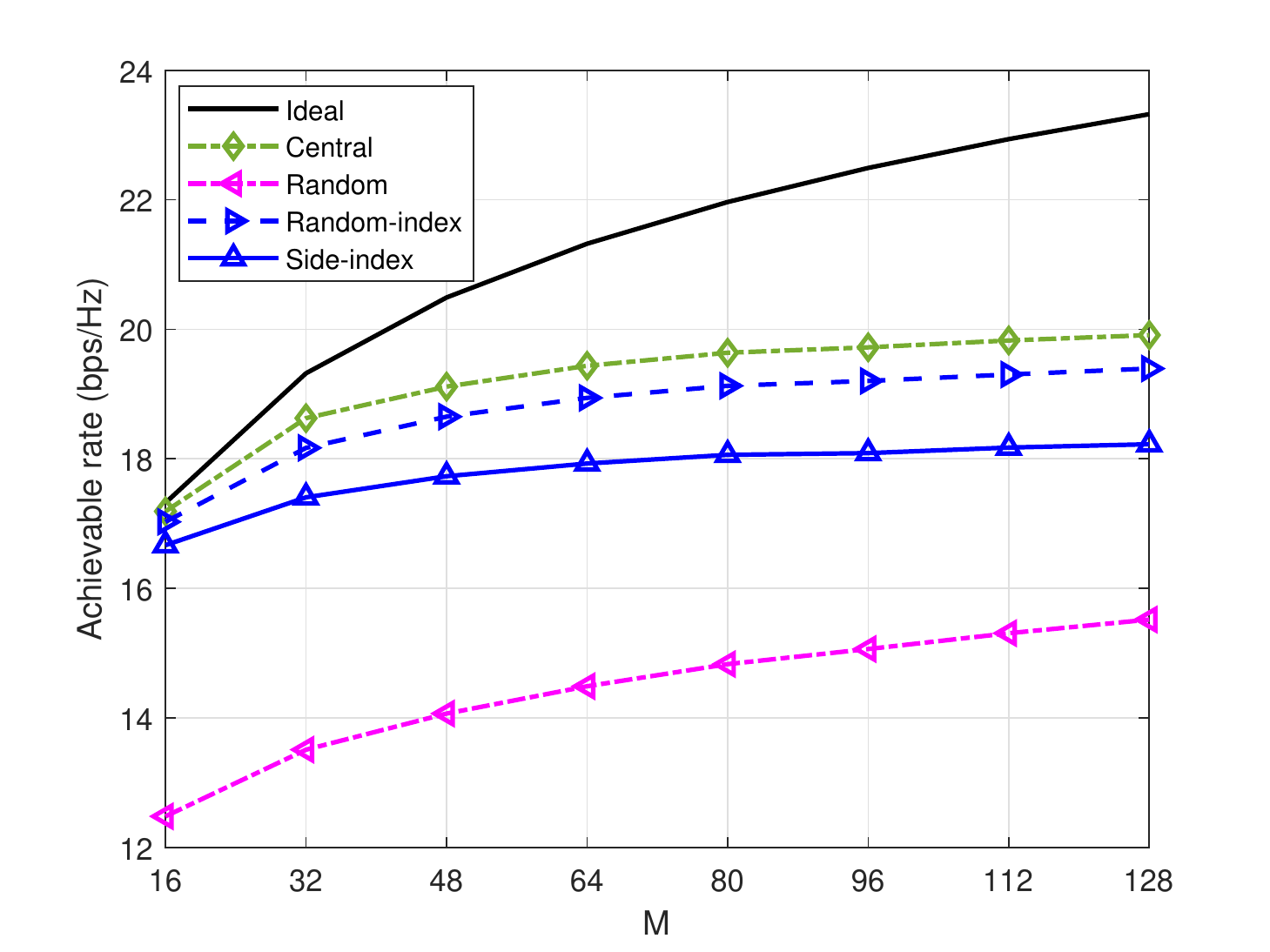}
\caption{Achievable rate vs the number of RIS elements in the LoS scenario.}
\label{Fig-M}
\end{figure}
Fig. \ref{Fig-M} illuminates the achievable rate of different schemes for different numbers of reflecting elements. It may be observed that the performance gaps between the upper bound and the other schemes increases with the number of reflecting elements, which demonstrates that the beam squint is more severe when deploying more reflecting elements in the RIS.

\subsection{NLoS scenario}
In this subsection, we evaluate the proposed MCCM based scheme for the NLoS scenario between the RIS and the user. Benchmarks similar to those in the LoS scenario are utilized to validate the superiority of the proposed scheme, where the difference is that `Central', `Random-index' and `Side-index' refer to the scheme based on the central subcarrier channel covariance matrix, the scheme based on the channel covariance matrix of the subcarrier whose index is randomly selected, and the schemes based on the channel covariance matrix of the subcarrier on either side, respectively. In this NLoS case, the number of paths is set to be 5.

\begin{figure}[t]
\centering
\includegraphics[scale=0.6]{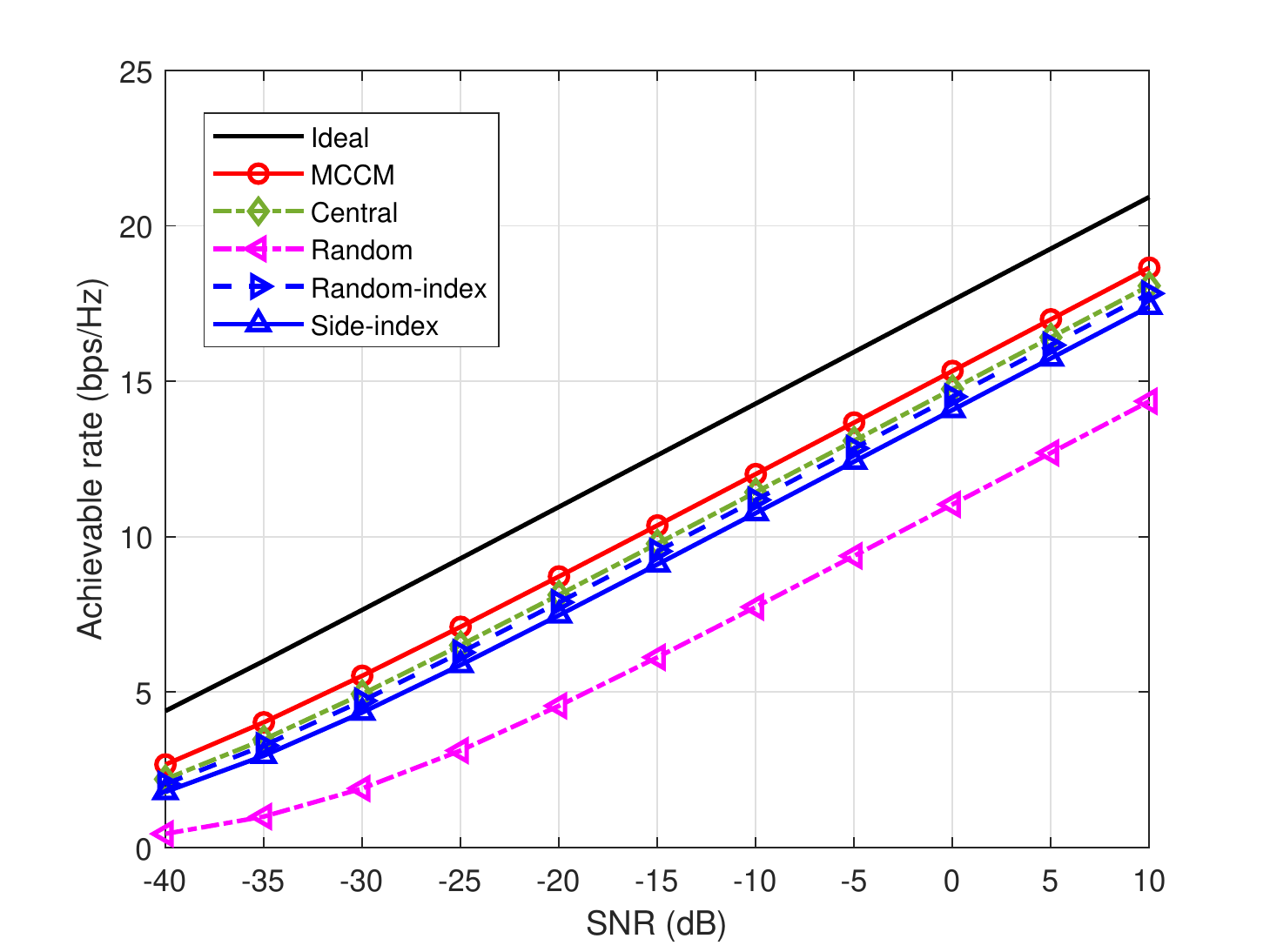}
\caption{Achievable rate vs SNR in the NLoS scenario.}
\label{Fig-SNR-NLoS}
\end{figure}
Fig. {\ref{Fig-SNR-NLoS}} compares the performance of different schemes for different SNRs in the NLoS scenario. We observe that the performance of the proposed MCCM based scheme is always better than the `Central', `Random-index' and `Side-index' schemes, which confirms that fully utilizing the correlations between both the paths and the subcarriers is able to effectively mitigate the influence of beam squint.

\begin{figure}[t]
\centering
\includegraphics[scale=0.6]{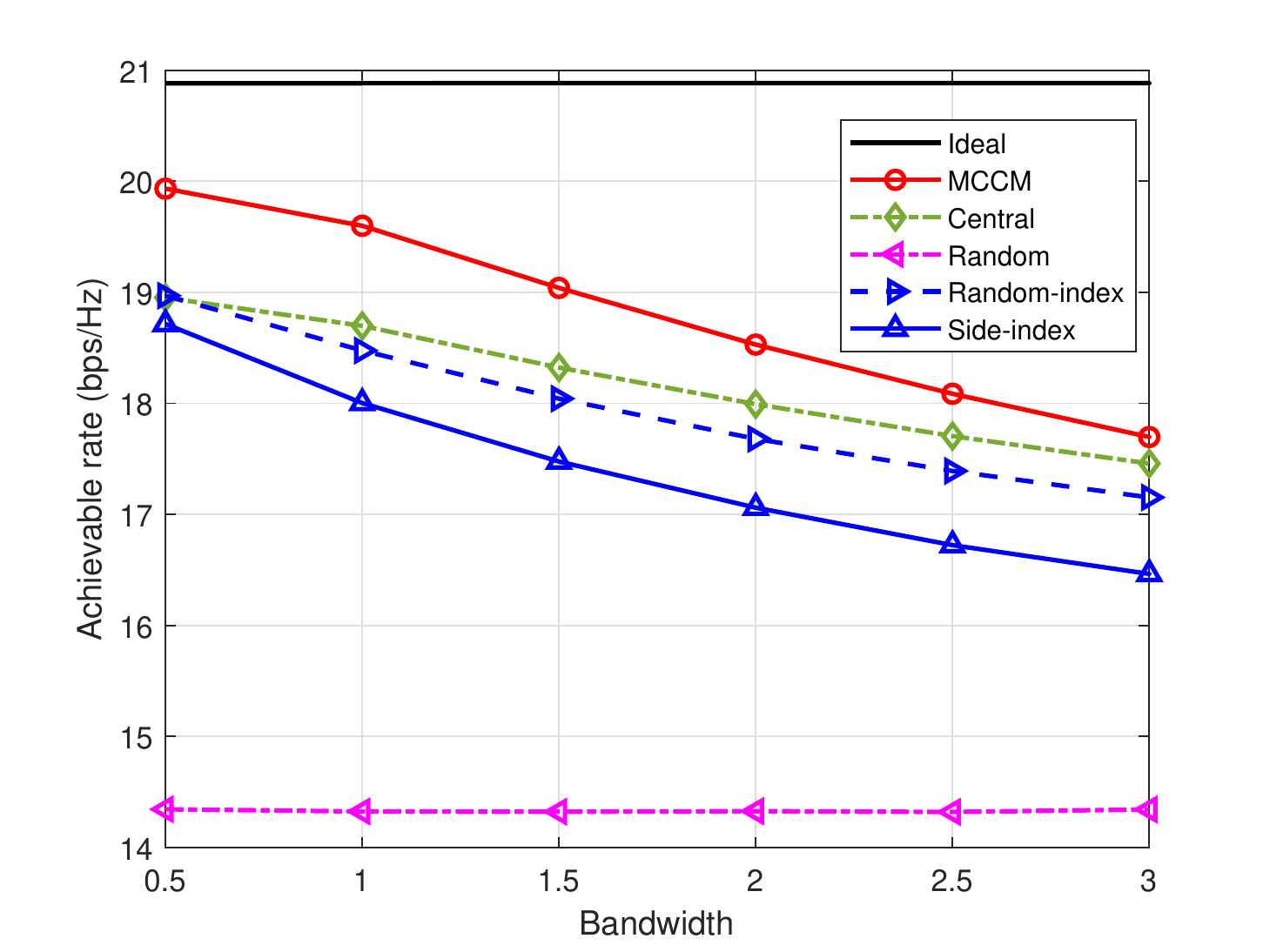}
\caption{Achievable rate vs Bandwidth in the NLoS scenario.}
\label{Fig-Bandwidth-NLoS}
\end{figure}
Finally, Fig. \ref{Fig-Bandwidth-NLoS} depicts the achievable rate of different schemes for different bandwidths in the NLoS scenario. It can be observed that even with 500 MHz bandwidth, there is still a performance loss close to 1 bps/Hz for the proposed scheme compared with the upper bound, while the other schemes loss even more than 2 bps/Hz in achievable rate. Therefore, the effect of beam squint is more intense in the NLoS scenario.

\section{Conclusions}
In this paper, several novel phase shift design schemes were proposed for mitigating the effect of beam squint in RIS assisted wideband mmWave communication systems. Specifically, for the LoS scenario between the RIS and the user, we firstly derived the optimal phase shift for each subcarrier. Then, to construct the common phase shift matrix for all subcarriers, we proposed a near-optimal scheme by maximizing the upper bound of the achievable rate, which is only based on the long-term angle information. Moreover, for the NLoS scenario, a MCCM based scheme was proposed by fully exploiting the correlations between both the paths and the subcarriers. Our extensive numerical experiments demonstrated that the beam squint would cause more than 3 bps/Hz performance loss when applying high bandwidth or deploying large number of elements in the RIS, as well as validated the effectiveness of the proposed schemes in mitigating the effect of beam squint.


\end{document}